\documentclass[letterpaper]{IEEEtran}
\usepackage{amsfonts}
\usepackage{amssymb}
\usepackage{amsmath}
\usepackage{graphicx}
\usepackage{subfigure}
\usepackage{cite}
\usepackage{setspace}
\usepackage{algorithm}
\usepackage{algorithmic}
\usepackage{multirow}
\usepackage{flushend}
\usepackage{amsmath}
\usepackage{xcolor}

\usepackage{dsfont}
\usepackage{amsmath,amsthm}
\usepackage{bm}
\usepackage{multicol}
\usepackage{amsfonts,amssymb}
\usepackage{array,color}
\usepackage{bbm}
\newtheorem{theorem}{$\mathbf{Theorem}$}
\newtheorem{lemma}[theorem]{$\mathbf{Lemma}$}

\begin{document}
\title{An Evolutionary Game for User Access Mode Selection in Fog Radio Access Networks}

\author{Shi Yan, Mugen Peng,~\IEEEmembership{Senior~Member,~IEEE,} Munzali Ahmed Abana, and Wenbo~Wang\\
{Key Laboratory of Universal Wireless Communication, Ministry of Education\\
Beijing University of Posts and Telecommunications, Beijing, 100876, China}}\maketitle

\begin{abstract}

The fog radio access network (F-RAN) is a promising
paradigm for the fifth generation wireless communication
systems to provide high spectral efficiency and energy efficiency.
Characterizing users to select an appropriate communication mode among fog access point (F-AP), and device-to-device (D2D) in F-RANs is critical for
performance optimization.
Using evolutionary game theory, we investigate the
dynamics of user access mode selection in F-RANs. Specifically, the competition among
groups of potential users space
is formulated as a dynamic evolutionary game, and the evolutionary
equilibrium is the solution to this game.
Stochastic geometry tool is used to derive the proposals' payoff expressions for both F-AP and D2D users by taking into account the different nodes locations, cache sizes as well as the delay cost.
The analytical results obtained from
the game model are evaluated via simulations, which show that the evolutionary game based access mode selection algorithm can reach a much higher payoff than the max rate based algorithm.

\end{abstract}

\section{INTRODUCTION}
To achieve the goals of the fifth generation (5G) mobile wireless systems \cite{C1} and alleviate the existing challenges in cloud radio access networks (C-RANs)\cite{C-RAN}$\sim$\cite{H-CRAN}, the fog radio access network (F-RAN) has been proposed as a new network architecture by incorporating of fog computing, edge storage
and centralized cloud computing into radio access networks \cite{F-RAN}\cite{F-RANzkc}.
Fog computing, which is similar to edge computing, is first proposed by Cisco \cite{Fog}. It extends cloud computing based services to the edge of the network. In F-RANs, services cannot only be executed in a centralized unit such as the base band unit (BBU) pool in C-RAN, but also can be hosted at smart terminal devices which are closer to the users. Meanwhile, through the user-centric adaptive techniques such as device-to-device (D2D), distributed coordination, and large-scale centralized
cooperation, users do not have to connect to the centralized cloud computing unit to complete the data transmission, which will improve both spectrum and energy efficiencies as well as relieve the load of fronthaul and alleviate the burden of BBU pool.
In order to execute the above, the traditional access point (AP) is evolved
to the fog access point (F-AP) through equipped with a certain caching and sufficient computing capabilities to execute the local cooperative signal processing in the physical layer.

Extensive studies have been done to study the user access mode schemes in C-RANs and heterogeneous cloud radio access networks (H-CRANs).
The performance of uplink transmissions under different user access modes in C-RANs have been investigated in \cite{b8}, and the impact of the remote radio heads (RRHs) intensity, the number of antennas and the pathloss exponent on the performance is characterized.
By using the stochastic geometry tool, the authors in \cite{b9} have provided
the probabilistic characterization of the signal interference ratio (SIR) distribution and the ergordic rate at a randomly located user in a downlink
H-CRAN system.
In \cite{b10}, the successful access probability has been
derived for C-RAN RRH cluster mode, and based on the analytical
result, two coalition game based cluster algorithms have
been designed. And in our previous study \cite{f}, the ergodic rates of two tier F-RANs under different access modes are characterized, and a mode selection mechanism is presented.

However, most of literatures focused on the mode selection schemes based on the user received SIR and distance, whereas little attention has been paid to study the impact of cache and fronthaul delay. The fronthaul limitations between RRHs and cloud are known to impose a formidable bottleneck to the system performances and a remarkable challenge to block the commercial practices \cite{Fronthaul}. Consequently, taking advantage of the fog-computing to switch the content cache or data process to network edge devices is the key to increase both of the spectral and energy efficiency, as well as relieve the burden of fronthaul in cloud computing based network architectures. D2D communications as one of the network edge communication technologies which has the ability of support content delivery and sharing is more and more paid extensive attention. Existing research on D2D networks mode selection mainly focused on underlaid cellular networks with fixed location model \cite{b12} $\sim$ \cite{b14}.
Different from the traditional underlaid cellular networks, the main challenge to design a user access mode selection algorithm in F-RANs is that the F-APs and D2D users are often deployed randomly and the network performance of F-RANs may drastically deteriorate with the increasing number of users who select the access mode with more fronthaul load.

In this paper, we present an adaptive user access mode selection algorithm for downlink F-RAN by taking into account the different nodes locations, cache sizes as well as the fronthaul delay cost.
The main contributions of this paper are three-folds.

\begin{itemize}
\item We characterize the ergodic rate and the coverage probability
of both D2D and
 F-AP modes in F-RANs. We investigate the impacts of the user node density, and the quality of
service (QoS) constrains and the intra-tier interference. The closed-form expressions are presented in some special cases,
which can make the analysis not only tractable, but also flexible.

\item Based on the above results obtained from the stochastic
geometry analysis, an evolutionary game theoretic approach
is presented to solve the problem of user access
mode selection, while taking into account different access
service requirements and fronthaul delay cost.

\item A distributed algorithm is proposed to reach the evolutionary
equilibrium, and its performance is compared with
that of the max rate based user access mode selection
scheme. The accuracy of the stochastic geometry analytical results are validated by Monte Carlo simulation results.
\end{itemize}

The remainder of this paper is organized as follows. Section II
describes the system model, cache model, delay model as well as the user access modes
in a downlink F-RAN. In Section III, the dynamic
behaviors of potential F-AP users are modeled as an evolutionary game. And we present the results of the coverage probability and ergodic rate in F-RAN system for both D2D and F-AP modes. While in Section IV, we propose
a distributed protocol to converge the evolutionary equilibrium, and
the evolutionary stable strategy (ESS) is analyzed. Simulation results provided in Section V prove that the performance comparisons of different
user access modes with different system parameters and the robustness of our protocol.
Finally, Section VI concludes the paper.

\section{SYSTEM MODEL}
\subsection{Radio Access Network Model}
As shown in Fig.1, a F-RAN downlink system is considered in this paper, where a group of F-APs are deployed according to a two-dimensional PPP $\Phi_f$
with density of $\lambda_f$ in a disc plane ${\cal{D}}^2$. Each F-AP is assumed as single antenna configuration with a fixed transmit power $P_f$, and interfaced to the cloud computing layer through a wired fronthaul network which is composed of links connecting F-APs with
several gateways. Then, the signals over gateways will be large-scale processed in the cloud network server via dedicated fiber \cite{Backhaul}. The locations of gateways are distributed according to a homogeneous PPP $\Phi_g$ with intensity $\lambda_g$.

The spatial distribution of users is
modeled as an independent PPP $\Phi_u$ distribution with constant intensity
$\lambda_u$. We set $p\sim(0,1]$ as the probability that a user can
direct transmit its local cache content to its nearby content require user (i.e.support D2D mode). Then, the location distribution of the D2D transmit users can be denoted as a
thinning homogeneous $\Phi_{tu}$ with the density of $\lambda_{tu}=p\lambda_u$.
Meanwhile according to Marking Theorem, the distribution of content require users follows a
stationary PPP $\Phi_{ru}$ with the density of $\lambda_{ru}=(1-p)\lambda_u$.

\begin{figure}[!htp]
\centering
\includegraphics[width=3.4in]{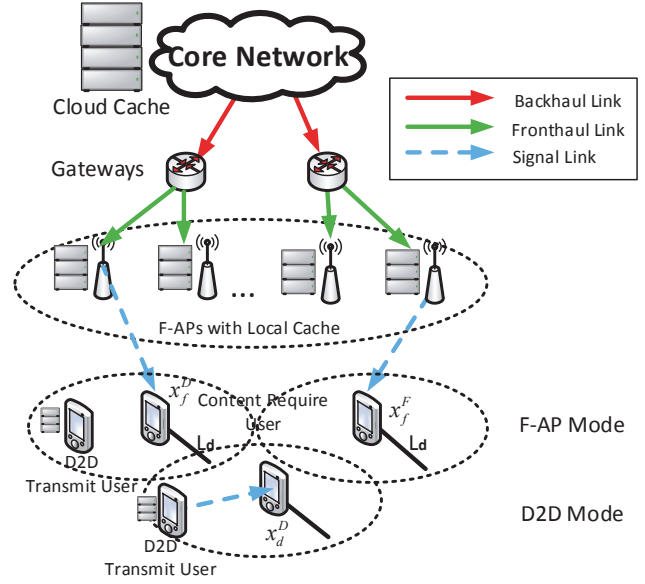}
\caption{Downlink transmission scenario}
\label{f1}\vspace*{-1em}
\end{figure}

In this paper, we focus on the user-centric access mode selection for content require users via the above F-RAN system within two phases, namely cache pre-fetching and access mode selection. The total number of content require users can be estimated as $N_{ru}= { {\cal A}({\cal{D}}^2)(1-p)\lambda_u}$, where ${\cal A}({\cal{D}}^2)$ is the area of the disc plane ${\cal{D}}^2$.

\subsection{Cache Pre-Fetching Phase}

In cache pre-fetching phase, the F-APs will cache content from the central cloud server file library which stored $N$ contents with uncoded cache strategies that may be requested by users $\mathcal {V}=\{v_1, v_2, ..., v_N\}$. Then, the D2D transmit users will cache part of the content from F-APs via wireless link. The pre-fetching content is limited by the local caching storage space size of F-APs and D2D transmit users, which denote by $C_f$ and $C_d$ ($C_d<C_f<N$), respectively, and the pre-fetching content in $C_f$ and $C_d$ is assumed to be constant across access mode selection phase.

Prior studies have been found that users are mostly interested in downloading
the most popular video contents \cite{Zipf}, which leads to only a small
portion of the $N$ contents are frequently accessed by the majority
of users. Therefore, in this paper we assume the F-APs and D2D transmit users only cache the most popularity content, and the demand probability can be modeled as the following Zipf distribution

\begin{equation}\label{popu}
{f_i(\sigma, N)} = \frac{{1/{i^\sigma }}}{{\sum\nolimits_{k = 1}^N {1/{k^\sigma }} }},
\end{equation}
where $\sum_{i=1}^N {f_i(\sigma, N)}=1$, and the video content with a smaller index has a larger probability of being requested by users, i.e. $f_i(\sigma, N)>f_j(\sigma, N)$, if $i<j$.
 Zipf exponent $\sigma>0$ controls
the relative popularity of files, and with the larger $\sigma$ the caching storage has a fewer of popular video contents accounting for the majority of the requests.

Next, we define the content caching probability as the probability of an event that a user $U$ can find the requested file $V$ in its accessed node, i.e., $p_c^x={\Pr }(V \in C_x)$. By setting the caching storages in F-APs and D2D transmit users to only cache the most popularly requested video contents, the content caching probabilities of each F-AP and D2D transmit user can be respectively obtained as

\begin{equation}\label{ccpf}
p^{F}_c={\Pr }(V \in C_f)={\sum\nolimits_{i=1}^{C_f}}{f_i(\sigma, N)},
\end{equation}
\begin{equation}\label{ccpm}
p^{D}_c={\Pr }(V \in C_d)={\sum\nolimits_{i=1}^{C_d}}{f_i(\sigma, N)}.
\end{equation}

\subsection{Access Mode Selection Phase}

In access mode selection phase, a content require user can select a communication mode independently according to the perceived performance, distance, local cache and the video streaming traffic delay cost. Let $U \to X$ signify that a desired content require user $U$ is associated to a node located at $X$, and ${\left\| X\right\|}$ denotes the distance between $U$ and $X$.
And we can categorize the content require users into two types: potential D2D mode users and potential F-AP mode users, the definitions of which are given as follows:

\begin{itemize}
\item \textbf{Potential D2D mode Users:}  A content require user is a potential D2D mode user if it can successfully obtain the requested contents from another D2D transmit user in a known location within a distance threshold $L_d$.
\end{itemize}

\begin{itemize}
\item \textbf{Potential F-AP mode Users:} A content require user is a potential F-AP mode user if the user cannot find a D2D transmit user within $L_d$, or the content require user can find a D2D transmit user near it but the requested content $V$ is not cached in that D2D transmit user. Thus, $U$ will try to access its nearest F-AP node.
\end{itemize}

It's worth noting that all potential D2D mode users can select both the D2D mode and F-AP mode, while each of the potential F-AP users only can select the F-AP mode to access its nearest F-AP.
More specifically, if a potential D2D mode user selects D2D mode, a communication link is established between the tagged content require user and its service D2D transmit user which has a known location. On the other hand, when F-AP mode is selected and the requested content $V$ is cached in the F-AP node, the potential D2D mode user or potential F-AP user can download data directly from F-AP. If user cannot find the requested content in its nearby F-AP, then user will download its data through the fronthaul in the same way as traditional C-RANs.

Let $N^{(D)}=N_{ru} p_d$ and $N^{(F)}=N_{ru}(1-p_d)$ denote the number of potential D2D mode and F-AP mode users, respectively, where $p_d$ denotes the probability of at least one D2D transmit user located in distance threshold $L_d$ meanwhile has the requested content $V$, which can be given as
\begin{equation}\label{D2Dp}
p_d=1-\exp({-\pi \lambda_{tu} p_c^D L_d^2}).
\end{equation}

\begin{proof}
By using the property of 2-D Poisson process, the probability distribution of the nodes number $m$ in a circle area $\pi l^2$ with radius limit $l$ can be derived as

\begin{equation}\label{PM}
{\rm{P_r}}\left\{ {\Phi \left( {\pi {l^2}} \right) = {m}} \right\}
 = \frac{{{{\left( {{\lambda_X}\pi {l^2}} \right)}^{{m}}}{e^{ - {\lambda_X}\pi {l^2}}}}}{{\left( {{m}} \right)!}}.
\end{equation}

Let $l=L_d$, $\lambda_X=p_c^D \lambda_{tu}$ and $m=0$. Then, we have the probability of none D2D transmit user has the the requested video content $V$ within the distance limit $L_d$. Therefore, \eqref{D2Dp} can be given as the probability of complementary events.
\end{proof}

During the access mode selection phase, we consider two kinds of delay, cache processing delay and average F-AP fronthaul packet delay. The cache processing delay mainly comes from the cache data processing time, and in this work we assume that the processing delay in D2D transmit users and F-APs as constants $\delta_d$ and $\delta_f$, respectively.

On the other hand, the F-AP fronthaul packet delay is defined as the time needed for a
packet to be transmitted from fronthaul to a F-AP. It depends on the size of the packet and the number of F-APs associated with the gateway \cite{frontahul}. In wired scenario, Gamma distribution is usually used to model the distribution of the delay in routers or switches \cite{frontahul1}\cite{frontahul2}. Then, the fronthual packet delay is expressed as

\begin{equation}\label{delayf}
T_f \sim \rm{Gamma}(\frac{\lambda_f}{\lambda_g}k, a+b\mu ),
\end{equation}
where the first and second terms of Gamma function represent the effect of number
of connecting nodes and packet size on delay, respectively. $a$,
$\mu$ and $k$ are constants, representing the processing capability
of the nodes, and $b$ is the packet size.

Therefore, we can directly obtain the mean F-AP fronthaul packet delay from the expectation of Gamma distribution as

\begin{equation}\label{meandelayf}
\overline{T_f}= \frac{\lambda_f}{\lambda_g}k(a+b\mu).
\end{equation}

And the total delay $\phi_i$ in access mode $i$ can be expressed as follows
\begin{equation}\label{totaldelay}
 \left\{ {\begin{array}{*{20}{l}}
  {\phi_d=p_c^D \delta_d},\\
  {\phi_f=p_c^F \delta_f+ (1-p_c^F)\overline{T_f}}.\\
\end{array}} \right.
\end{equation}

\section{EVOLUTIONARY GAME APPROACH for USER ACCESS MODE SELECTION}

In this section, the evolutionary game formulation
for user access mode selection is given, and the replicator dynamics is used to model the strategy adaptation process.

\subsection{Formulation of the Evolutionary Game}
We define the evolutionary game as follows:

\begin{itemize}
\item \textbf{Players:} Each user who can choose among multiple access mode is a player of the game. In this paper, the players are all of the potential D2D mode users. Note that the potential F-AP mode users are not involved in the game since F-AP mode is the only access mode available to these users.

\item \textbf{Strategies:} The set of strategies is the selection of a communication mode. The set of strategies for the potential D2D mode players is $\mathcal {D}=\{D, F\}$, and there is obvious only one strategy possible for the potential F-AP players $\mathcal {F}=\{F\}$.

\item \textbf{Payoff:} The payoff quantifies the performance satisfaction
    level of a potential player which depends on the ergodic rate
    rate as well as the delay cost. The payoff function for a $(n)$ mode potential player selected access mode $i$ is defined as follows:
\end{itemize}

\begin{equation}\label{payoff}
\begin{gathered}
\pi_i^{(n)}(\textbf{x})=\frac{p_iB_i}{\Sigma_{a\in A^{(n)}}n_i^{(a)}}C_i^{(n)}-q_i\phi_i,\hfill\\
\quad\quad\quad=\frac{p_iB_i}{\Sigma_{a\in A^{(n)}}N^{(a)}x_i^{(a)}}C_i^{(n)}-q_i\phi_i,\hfill\\
\end{gathered}
\end{equation}
where \textbf{x} denotes the vector of the proportion of users choosing different access modes. $C_i^{(n)}$ denotes the utility function measuring the achieved performance in $(n)$ mode choosing mode $i$, which will be further discussed in Section IV. $B_i$ denotes the bandwidth in mode $i$. $x_i^{(a)} = n_i^{(a)}/N^{(a)}$ is the proportion of individuals of potential users in $(a)$ choosing strategy $i$. $n_i^{(a)}$ is the number of potential users in $(a)$ mode choosing mode $i$, $A^{(n)}$ is the set of subareas
in mode $n$, for the potential D2D mode users, this
set can be defined as $A^{(n)} = \{D, F\}$ since both D2D and F-AP modes are available for these users. $p_i$ is the coefficient of linear pricing function used by mode $i$ to charge a user.

Moreover, since the F-AP fronthaul bandwidth is limited, it will lead to negative impact on their payoff if too many users select F-AP mode to receive their data through fronthaul link. On the other hand, the processing delay cost of both of D2D and F-AP modes will also increase with the growth of the users at the same time. In order to evaluate these impacts, we define the cost price coefficient $q_d$ and $q_f$ to reflect the influence of the proportion of mode selection.

\begin{equation}\label{qi}
 \left\{ {\begin{array}{*{20}{l}}
   {q_{d}=c_1\left[\exp(c_2N^{(D)}x^{(D)}_d)-1\right],}\\
   {q_{f}=c_3\left[\exp\left(c_4\left(N^{(F)}+N^{(D)}x^{(D)}_f\right)\right)-1\right],}\\
\end{array}} \right.
\end{equation}
where $c_1, c_2, c_3$ and $c_4$ are positive constants.

In order to measure the users' performance utility functions presented in the evolutionary game formulation \eqref{payoff}, in the next subsection, we use stochastic geometry tool to derive the coverage probability and ergodic rate for both D2D and F-AP modes by taking into account the different nodes locations, the SIR QoS constrains and the interference.

\subsection{Performance Analysis for Different User Access Modes}

In this subsection, we derive the coverage probability and ergodic rate for F-RAN with two different user access modes. The ergodic rate in F-RAN is defined as $R_i=\mathbb{E}[\rm{ln}\left(1+SIR(U \to X_i)\right)|SIR(U $$ \to X_i)>T_i]$, where the unit of the ergodic rate is in terms of nats/s/Hz, $T_i$ is the QoS constraint of access mode $i$.

In this paper, we limit our attention to the interference-limited scenario since the interference is much larger than the noise.
Path loss is represented by a standard power
law path loss function ${\left\| X \right\|^{ - \alpha }}$, where $\alpha>2$ is the path loss exponent.

Then, if $U$ is served by a D2D transmit user at a fixed distance of $\left\| {{X_f}} \right\|$ to $U$, the received SIR at the desired user is given by

\begin{equation}\label{SIRd}
SIR(U \to {X_d}){\rm{ }} = \gamma_d= \frac{{{P_d}{h_d}{{\left\| {{X_d}} \right\|}^{ - \alpha}}}}{{{I_{d,ru}} + {I_{f,ru}}}},
\end{equation}
where $h_d$ characterizes the flat Rayleigh channel fading between the potential D2D mode user and its associated D2D transmit user, and ${{\left\| {{X_d}} \right\|}^{ - \alpha }}$ denotes the path loss.
${I_{d,ru}} = {\sum _{i \in {\Phi _{tu}}}}{P_d}{g_i}r_i^{ - \alpha }$ is the intra-tier interference from other D2D transmit users, $ g_i \sim \exp (1)$ and $r_i^{ - \alpha }$ denote the exponentially distributed fading power over the Rayleigh fading channel and path loss from other D2D transmit users to $U$, respectively. ${I_{d,ru}} = {\sum _{j \in {\Phi _f}}}{P_f}{g_j}l_j^{ - \alpha }$ denotes inter-tier interference from F-APs, the definition of ${g_j}$ and $l_j^{ - \alpha }$ are similar to that in ${I_{d,ru}}$.

Next, if $U$ is served by a single nearest F-AP in F-AP mode, the SIR is given by

\begin{equation}\label{SIRf}
SIR(U \to {X_f}){\rm{ }} = \gamma_f= \frac{{{P_f}{h_f}{{\left\| {{X_f}} \right\|}^{ - \alpha }}}}{{{I_{d,fu}} + {I_{f,fu}}}},
\end{equation}
where ${I_{f,fu}} = {\sum _{{j} \in {\Phi _f}/f}}{P_f}{g_j'}l_{j}^{' - \alpha }$, and ${I_{d,fu}} = {\sum _{{j'} \in {\Phi _{tu}}}}{P_d}{g_{j'}}r_{j'}^{ - \alpha }$ denote the intra-tier interference from other F-APs and inter-tier interference from D2D transmit users, respectively, and the definitions of $h_f$, $g_i'$ and $l_i'$ are given in \eqref{SIRd}.

\textbf{1.D2D mode}

The coverage probability of a user in D2D mode can be denoted as

\begin{equation}\label{PD}
\begin{gathered}
P_d(T_d, \alpha, \left\| {{X_d}} \right\|) = \Pr \left( {\frac{{{P_d}{h_d}{{\left\| {{X_d}} \right\|}^{ - {\alpha}}}}}{{{I_{tu,ru}} + {I_{f,du}}}} \ge {T_d}} \right)\hfill\\
\quad\quad\quad=\exp \left( { - \pi {{\left\| {{X_d}} \right\|}^2}\left( {{p_c^{D}\lambda _{tu}} + {{\left( {\frac{P_f}{P_d}} \right)}^{\frac{2}{{{\alpha}}}}}{\lambda _{f}}} \right)C\left( {{\alpha}} \right)T_d^{\frac{2}{{{\alpha}}}}} \right),\hfill\\
\end{gathered}
\end{equation}
where $C\left( \alpha  \right) = {{2{\pi}\csc \left( {{{2\pi }}/{\alpha }} \right)}}/{\alpha }$.
\begin{proof}
See Appendix A.
\end{proof}

Similarly, the achievable ergodic rate can be derived as
 \begin{equation}\label{RD}
\begin{gathered}
 {R_d} = \mathbb{E} \left[ {\ln \left( {1 + {\gamma _d}} \right)} | \gamma_d \ge T_d\right] \hfill\\
 \mathop  \approx \limits^{\left( a \right)}  \ln(T_d)P_d(T_d,\alpha \left\| {{X_d}} \right\|) - \frac{{{\alpha}}}{2}\hfill\\
\cdot{\rm{Ei}}\left[ { -T_d^{\frac{2}{\alpha}} \pi {{\left\| {{X_d}} \right\|}^{2}}\left( {{p_c^{D}\lambda _{tu}} + {{\left( {\frac{P_f}{P_d}} \right)}^2}{\lambda _f}} \right)C\left( {{\alpha}} \right)} \right], \hfill\\
\end{gathered}
 \end{equation}
where $(a)$ follows in the high SIR conditions ${\ln}\left( {1 + \gamma_d} \right) \to {\ln }\left( \gamma_d \right)$, Ei$[s]=-\int_{-s}^\infty{{e^{-t}}/{t}\rm{d}t}$ is the exponential integral function.
\begin{proof}
See Appendix B.
\end{proof}

\textbf{2.F-AP mode}

In F-AP mode, the desired user will try to access its nearest F-AP $X_f$ to download the requested content $V$. The probability density function (PDF) of the distance between $X_f$ and $U$ can be derived by using a similar way as \eqref{D2Dp}
 \begin{equation}\label{cdfX}
\begin{gathered}
 {f_{\left\| {{X_f}} \right\|}}\left( {{r_f}} \right) = \frac{{\partial \left( {1 - \Pr \left( {{\rm{No\;F-AP\;closer\;than\; }}{r_m}} \right)} \right)}}{{\partial {r_f}}} \hfill\\
  = \frac{{\partial \left( {1 - \exp \left( { - \pi {\lambda _f}r_f^2} \right)} \right)}}{{\partial {r_f}}} = 2\pi {\lambda _f}{r_f}{e^{ - \pi {\lambda _f}r_f^2}}. \hfill\\
 \end{gathered}\
 \end{equation}

 Thus, the coverage probability of F-AP mode can be calculated as
  \begin{equation}\label{PF}
  {P_f}\left( {{T_f},\alpha} \right) =\frac{1}{{1 + \rho \left( {{T_f},{\alpha}} \right) + \frac{{{\lambda _{tu}}}}{{{\lambda _f}}}C\left( {{\alpha}} \right){{\left( {\frac{{{P_d}{T_f}}}{{{P_f}}}} \right)}^{2/{\alpha}}}}},
  \end{equation}
 where $\rho \left( {T_f,\alpha } \right) = \int_{{T_f^{ - \frac{2}{\alpha}}}}^\infty  {\frac{{{T_f^{2/\alpha }}}}{{1 + {v^{\alpha /2}}}}{\rm{d}}v}$.
\begin{proof}
See Appendix C.
\end{proof}

And the ergodic rate for F-AP mode can be given as
 \begin{equation}\label{R1F}
  \begin{gathered}
   {R_f} =  \mathbb{E} \left[ {\ln \left( {1 + {\gamma _f}} \right)} | \gamma_f \ge T_f\right] \hfill\\
\approx  \int_{\ln(T_f)}^{\infty}   {P_f(e^\theta, \alpha )} {\rm{d}}\theta+ \ln(T_f)P_f(T_f, \alpha ). \hfill\\
  \end{gathered}
 \end{equation}

\textbf{Special Case}: {Path loss exponent for F-AP to user link is 4 ($\alpha=4$), and SIR threshold $T_f > 1$}

An approximated closed-form approximate expression can be derived in this special case, and we give the ergodic rate in the following lemma 1.

\begin{lemma}
The ergodic rate for nearest F-AP mode with $\alpha$= 4 and $T_f > 1$ can be expressed as
\begin{equation}\label{R1FS}
  \begin{gathered}
   {R_f}^{\alpha=4} =  \mathbb{E} \left[ {\ln \left( {1 + {\gamma _f}} \right)} | \gamma_f \ge T_f\right] \hfill\\
\approx  \frac{4 }{\pi\sqrt{T_f}\left(1+\frac{\lambda_{tu}}{\lambda_f}\sqrt{\frac{P_d}{P_f}}\right)}
+\frac{2 \ln(T_f)}{\pi\sqrt{T_f}\left(1+\frac{\lambda_{tu}}{ \lambda_f}\sqrt{\frac{P_d}{P_f}}\right)} \hfill\\
=\frac{2 (2+\ln(T_f))}{\pi\sqrt{T_f}\left(1+\frac{\lambda_{tu}}{ \lambda_f}\sqrt{\frac{P_d}{P_f}}\right)}.\hfill\\
  \end{gathered}
 \end{equation}
\end{lemma}

\begin{proof}
See Appendix D.
\end{proof}

\subsection{Replicator Dynamics}

Evolutionary game theory differs from classical game theory by focusing more on the dynamics of strategy change as influenced not solely by the quality of the various competing strategies, but by the effect of the frequency with which those various competing strategies are found in the population \cite{Evolutionary}\cite{EvolutionaryGame}.
In a dynamic evolutionary game, the replicator dynamic is a simple model of
evolution and prestige-biased learning in games \cite{Evolutionary1}.

We consider an evolutionary game of user access mode selection where the group of potential users in $(a)$ can choose among the available wireless access modes $i$.
The game is repeated, and the user adopts a strategy that gives a higher payoff which is determined by the payoff matrix and the proportion of each strategy in
the population. The speed
of the user in observing and adapting the access mode selection is
controlled by parameter $\rho>0$.

For a small period of time, the rate of strategy
is governed by the replicator dynamics, which is defined as
follows:

\begin{equation}\label{replicator}
\dot{x}_i^{(n)}=\rho x_i^{(n)}\left(\pi_i^{(n)}(\textbf{x})-\overline{\pi}^{(n)}(\textbf{x})\right),
\end{equation}
where $\overline{\pi}^{(n)}(\textbf{x})$ is the average
payoff of the users in potential group $(n)$ which is computed as

\begin{equation}\label{average payoff}
\overline{\pi}^{(n)}(\textbf{x})=\Sigma_ix_i^{(n)}\pi_i^{(n)}(\textbf{x}).
\end{equation}

Based on this replicator dynamics of the users
in potential group $(n)$, the number of users choosing access mode $i$ increases if
their payoff is above the average payoff. It is impossible for
a user to choose access mode $k$, which provides a lower payoff
than the current payoff. This replicator dynamics satisfies the
condition of $\Sigma_i\dot{x}_i^{(n)}=1$.
In the next section, we will use above expressions to analyze the stability of the access mode selection game.

\section{EVOLUTIONARY EQUILIBRIUM and STABILITY ANALYSIS}

In this section, the evolutionary
equilibrium is considered to be the solution of the formulated access mode selection game and we also analyzed its stability.
An evolutionary equilibrium is a
fixed point where there is no change in proportion of
players choosing different strategies. In other words, since the rate of strategy adaptation
is zero (i.e., $\dot{x}_i^{(n) }= 0$), there is no user who deviates to gain a
higher payoff \cite{Evolutionary2}.

To evaluate the stability at the fixed point $\dot{x}_i^{(n)*}$ , which is
obtained by solving $\dot{x}_i^{(n) }= 0$, the eigenvalues of the Jacobian
matrix corresponding to the replicator dynamics need to be
evaluated. The fixed point is stable if all eigenvalues have a
negative real part \cite{Evolutionary3}.

For the mode selection game in this paper, the replicator dynamics
can be expressed as follows:

\begin{equation}\label{replicatorSC}
\begin{gathered}
\dot{x}_f^{(D)}=\sigma x_f^{(D)}\left(\pi_f^{(D)}
(\textbf{x})-x_f^{(D)}\pi_f^{(D)}(\textbf{x})-x_d^{(D)}\pi_d^{(D)}(\textbf{x})\right)\hfill\\
\quad\quad\mathop = \limits^{\left( a \right)} \sigma x_f^{(D)}\left(\left(1-x_f^{(D)}\right)\pi_f^{(D)}(\textbf{x})-\left(1-x_f^{(D)}\right)\pi_d^{(D)}(\textbf{x})\right)\hfill\\
\quad\quad=\sigma x_f^{(D)}(1-x_f^{(D)})\left(\pi_f^{(D)}(\textbf{x})-\pi_d^{(D)}(\textbf{x})\right),\hfill\\
\end{gathered}
\end{equation}
where $(a)$ follows that $x_f^{(D)}+x_d^{(D)}=1$.

Hence, we could derive the evolution
algorithm as described in Algorithm 1. Also, the stability of the equilibrium is clarified in Theorem 2.

\begin{algorithm}[]\tiny
   \caption{Evolution Algorithm}
   \begin{algorithmic}[1]
   \small\STATE \textbf{Initialize}  Each user selects to get access to a access mode in random.\\
    \STATE \textbf{Step 1} The payoff of users accessing $i$ mode $\pi_i^{(n)}$ is calculated according to \eqref{payoff}. This payoff information is sent to the F-APs.\\
   \STATE \textbf{Step 2} The average payoff of the content require users $\overline{\pi}^{(n)}$ is calculated according to \eqref{average payoff} in F-APs, and is broadcast to all potential D2D mode users. \\
    \STATE \textbf{Step 3} For a user accessing mode $i$, if its payoff is less than the average payoff, i.e., $\pi_i^{(n)}<\overline{\pi}^{(n)}$, it would randomly
   switch to another access mode $j$, where $j\neq i$ and $\pi_j^{(n)}>{\pi}_i^{(n)}$\\
   \STATE Repeat from Step 1 to Step 3 until convergence.
   \end{algorithmic}
\end{algorithm}
\begin{theorem}
For a potential D2D mode user with a stable equilibrium point, the interior evolutionary equilibrium in game is asymptotically stable.
\end{theorem}

\begin{proof}
See Appendix E.
\end{proof}

\subsection{Comparison with Max Rate Based Access Mode Selection Algorithm}
For comparison purpose, we consider a max rate based access mode selection algorithm under the same system model. Max rate based access mode selection algorithm needs to sort both the D2D mode and F-AP mode rate list based on QuickSort algorithms to have ordered FAP and MBS rate lists, and then control each user to access to its node with max rate. We derive the max rate based access mode selectio
algorithm as described in Algorithm 2.

\begin{algorithm}[!htp]\tiny
   \caption{Max Rate Based User Access Mode Selection Algorithm}
   \begin{algorithmic}[1]
   \small\STATE \textbf{Initialize}  Each user selects to get access to a access mode in random.\\
    \STATE \textbf{for} $u= 1,2,3,...,N_u$ do. \\
    \STATE \textbf{Step 1} Find the max rate D2D transmit user with the required content $V_u$ of the desired user $U_u$ within a radius threshold $L_f$, $B(U_u,L_f)\cap\Phi_{tu}=\{X_1,X_2,...,X_D\}$.\\
    \STATE \quad\textbf{if} $B(U_u,L_f)\cap\Phi_{tu}= \phi$.\\
    \STATE \quad\quad $R_d^{max}=0$,\quad\textbf{go to Step 2}.\\
    \STATE \quad\textbf{end if}\\
    \STATE \quad Calculate the D2D $X_1$ rate $R_{X_1}$ from \eqref{RD}.\\
    \STATE  \quad\textbf{Set} $R_d^{max}=R_{X_1}$.
   \STATE \quad\quad\textbf{for} $i= 2,3,...,D$ do. \\
   \STATE \quad\quad\quad\textbf{if} $R_{X_i}>R_d^{max}$ .\\
   \STATE \quad\quad\quad Max Rate D2D transmit user $R_d^{max}=R_{X_i}$ .\\
   \STATE \quad\quad\quad\textbf{end if}\\
   \STATE \quad\quad\textbf{end for}\\
      \STATE \textbf{Step 2} Find the max rate F-AP of the desired user $U_u$ in $\Phi_{f}=\{Y_1,Y_2,...,Y_F\}$.\\
    \STATE \quad Calculate the F-AP $Y_1$ rate $R_{Y_1}$ from \eqref{R1F}.\\
    \STATE  \quad\textbf{Set} $R_f^{max}=R_{Y_1}$.
   \STATE \quad\quad\textbf{for} $j= 2,3,...,F$ do. \\
   \STATE \quad\quad\quad\textbf{if} $R_{Y_j}>R_f^{max}$ .\\
   \STATE \quad\quad\quad Max Rate F-AP node $R_f^{max}=R_{Y_j}$ .\\
   \STATE \quad\quad\quad\textbf{end if}\\
   \STATE \quad\quad\textbf{end for}\\
  \STATE \textbf{Step 3} \\
  \STATE \quad\quad \textbf{if} $R_d^{max}>R_f^{max}$ .\\
   \STATE \quad\quad \quad User $U_u$ select \textbf{D2D mode}.\\
   \STATE \quad\quad \textbf{else}\\
   \STATE \quad\quad \quad User $U_u$ select \textbf{F-AP mode}.
   \STATE \quad\quad\textbf{end if} \\
  \STATE \textbf{end for}\\
     \end{algorithmic}
\end{algorithm}

\section{NUMERICAL RESULTS}

In this section, the accuracy of the above ergodic rate expressions
and the performance of the proposed algorithm are
evaluated by using Matlab with Monte Carlo simulation method.
The simulation parameters are listed as follows in Table I, and the simulation network topology is shown in Fig. 3.

\begin{table}[!htp]
\renewcommand{\arraystretch}{1.3}
\caption{SIMULATION PARAMETERS} \label{table_example}
\centering
\begin{tabular}{{|c|c|c|}}
\hline
\bfseries Parameters & \bfseries Value &\bfseries Description\\
\hline
$N$& $1000$ & Number of video contents\\
\hline
$L_d$ & 30m  & D2D distance threshold\\
\hline
$C_d$& $80$ &Cache size of D2D transmit users\\
\hline
$C_f$ & $100 \sim 500$& Cache size of F-APs \\
\hline
$B_d$ & 300MHz & Bandwidth of D2D users\\
\hline
$B_f$ & 100MHz & Bandwidth of F-APs\\
\hline
 $P_d$ & 13dBm & Transmit power of D2D transmit user\\
\hline
 $P_f$ & 23dBm & Transmit power of F-AP\\
\hline
$\lambda_{ru}$ & $6 \times 10^{-5}$ &Intensity of content require users\\
\hline
$\lambda_{tu}$ & $5 \times 10^{-5}$ &Intensity of D2D transmit users\\
\hline
$\lambda_f$ & $3\times 10^{-5}$ &Intensity of F-APs\\
\hline
$\lambda_g$ & $2\times 10^{-6}$ &Intensity of gateways\\
\hline
$\alpha$  & 4  \cite{tc} & Path loss exponent\\
\hline

\end{tabular}\vspace*{-1em}
\end{table}

\begin{figure}[!htp]
\centering
\includegraphics[width=3.4in]{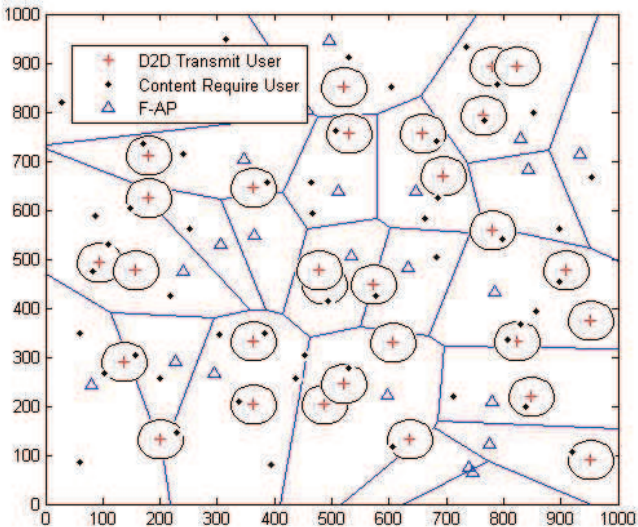}
\caption{Poisson distributed content require users, D2D transmit users and F-APs. The F-APs coverage boundaries are shown and form a Voronoi tessellation.}
\label{f2}\vspace*{-1em}
\end{figure}

\begin{figure}[!htp]
\centering
\includegraphics[width=3.4in]{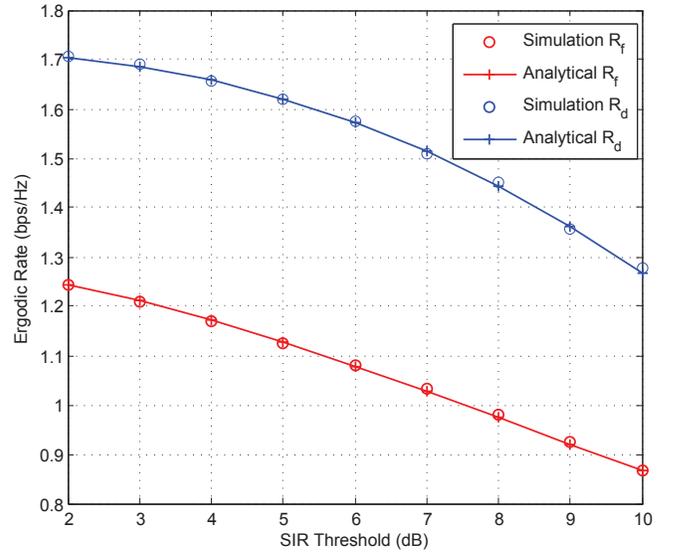}
\caption{Ergodic rate of D2D mode and F-AP mode versus different SIR thresholds.}
\label{f3}\vspace*{-1em}
\end{figure}

Fig. 4 shows the ergodic rate achieved by the D2D mode and F-AP mode in F-RAN with the varying SIR QoS thresholds $T$.
The analytical results closely match with the
corresponding simulation results, which validates our analysis
in Section III. It can be observed that the ergodic rate of both D2D and F-AP user access mode decreases as the SIR threshold increases. This is because the larger $T$ suggests that the user is more strict in the quality of SIR, which leads to ergodic rate decreasing.

\begin{figure}[!htp]
\centering
\includegraphics[width=3.4in]{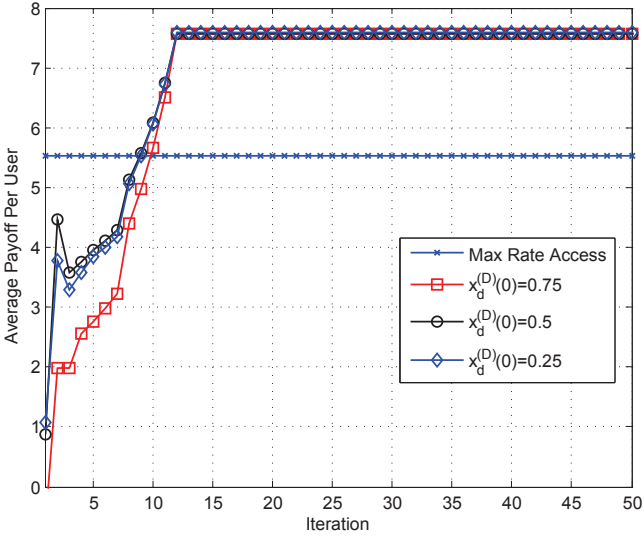}
\caption{Average payoff per user with the user access mode selection algorithm.}
\label{f3}\vspace*{-1em}
\end{figure}

\begin{figure}[!htp]
\centering
\includegraphics[width=3.4in]{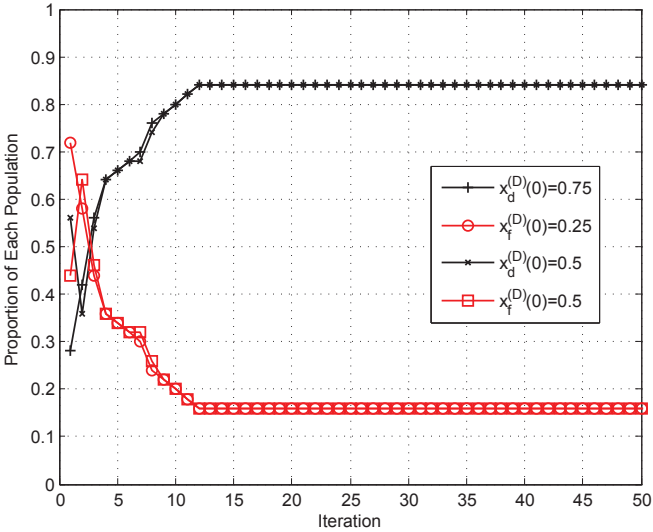}
\caption{Proportion of user with the user access mode selection algorithm.}
\label{f3}\vspace*{-1em}
\end{figure}

The proportion and the achieved
average payoff per user accessing different mode during the
evolution process are shown in Fig. 5 and Fig. 6. It can be seen that the proposed user access mode selection algorithm can reach the evolutionary equilibrium very fast, within less than 15 iterations.
More specifically, as can be seen from Fig. 5, we
could find that after several iterations all curves with different proportion at the initial point at last achieve a much better average payoff than only consider max rate based user access mode selection algorithm.

On the other hand, the trajectories of the proportion of users selecting each
access mode are illustrated in Fig. 6. It can be observed that the
proportions eventually converge to the equilibrium after changing in different directions. At the same time, those
curves about payoff finally converge to the network average payoff.

\section{CONCLUSION}
In this paper, an evolutionary game theory based algorithm
to solve the problem of user access mode selection
in downlink F-RAN is proposed. Stochastic geometry is used to derive the payoff expressions by taking into account the different nodes locations, interferes as well as the delay cost. The evolutionary equilibrium has been considered to
be the stable solution for which all users receive identical payoff from accessing different access modes. The simulation results show that the evolutionary game based access mode selection algorithm can reach a much higher payoff than the max rate based algorithm.

\appendix
\subsection{Proof of $P_d(T_d, \alpha, \left\| {{X_d}} \right\|)$}
According to the definition of coverage probability, we have
\begin{equation}\label{PFproof}
\begin{gathered}
 P_d(T_d, \alpha, \left\| {{X_d}} \right\|) = \Pr \left( {\frac{{{P_d}{h_d}{{\left\| {{X_d}} \right\|}^{ - {\alpha}}}}}{{{I_{d,ru}} + {I_{f,ru}}}} \ge {T_d}} \right) \hfill\\
  = \Pr \left( {{h_d} \ge \frac{{{T_d}{{\left\| {{X_d}} \right\|}^{  {\alpha }}}}}{{{P_d}}}\left( {{I_{d,ru}} + {I_{f,ru}}} \right)} \right) \hfill\\
 \mathop  = \limits^{\left( a \right)} \mathbb{E}\left[ {\exp \left( { - \frac{{{T_d}{{\left\| {{X_d}} \right\|}^{  {\alpha}}}}}{{{P_d}}}\left( {{I_{d,ru}} + {I_{f,ru}}} \right)} \right)} \right] \hfill\\
 \mathop  = \limits^{\left( b \right)} {{\cal L}_{{I_{d,ru}}}}\left( {\frac{{{T_d}{{\left\| {{X_d}} \right\|}^{  {\alpha}}}}}{{{P_d}}}} \right){{\cal L}_{{I_{d,ru}}}}\left( {\frac{{{T_d}{{\left\| {{X_d}} \right\|}^{  {\alpha}}}}}{{{P_d}}}} \right) \hfill\\
=\exp \left( { - \pi {{\left\| {{X_d}} \right\|}^2}\left( {{p_c^{D}\lambda _{tu}} + {{\left( {\frac{P_f}{P_d}} \right)}^{\frac{2}{{{\alpha}}}}}{\lambda _f}} \right)C\left( {{\alpha}} \right)T_d^{\frac{2}{{{\alpha}}}}} \right), \hfill\\
 \end{gathered}
 \end{equation}
where (a) follows from the Laplace transform of $h_d \sim \exp(1)$ and the independence of $I_{d,ru}$ and $I_{f,ru}$ \cite{La}\cite{La2}. (b) follows from letting $s = {{{{T_d}{{\left\| {{X_d}} \right\|}^{  {\alpha}}}}}/{{{P_d}}}}$ in the Laplace transforms of $I_{d,ru}$ and $I_{f,ru}$, which are given by

\begin{equation}\label{lfu}
{{\cal L}_{{I_{d,ru}}}}\left( s \right) = {{{\rm{E}}_{{I_{d,ru}}}}\left[ {\prod\limits_{{i} \in {\Phi _{tu}}} {\exp \left( { - s{P_d}{g_{{i}}}r_i^{-\alpha}} \right)} } \right]}.
\end{equation}

\begin{equation}\label{lmu}
 {{\cal L}_{{I_{f,ru}}}}\left( s \right) = {{{\rm{E}}_{{I_{f,ru}}}}\left[ {\prod\limits_{{j} \in {\Phi _f}} {\exp \left( { - s{P_f}{g_{{j}}}l_j^{-\alpha}} \right)} } \right]}.
\end{equation}

Using the independence of the fading random variables in ${\cal L}_{{I_{{d,ru}}}}(s)$ ,we have

\begin{equation}\label{Pr}
\begin{gathered}
{{\cal L}_{{I_{d,ru}}}}\left( s \right) = {{{\rm{E}}_{{I_{d,ru}}}}\left[ {\prod\limits_{{i} \in {\Phi _{tu}}} {\exp \left( { - s{P_d}{g_{{i}}}r_i^{-\alpha}} \right)} } \right]}\hfill\\
 \mathop  = \limits^{\left( a \right)} \left[ {\prod\limits_{{i} \in {\Phi _{tu}}} {\frac{1}{{1 + s{P_d}r_i^{-\alpha}}}} } \right]  \hfill\\
 \mathop  = \limits^{\left( b \right)}  {\exp \left( { - {p_c^{D}\lambda _{tu}}\int_{{{\rm{r}}^2}} {\left( {1 - \frac{1}{{1 + s{P_d}r_i^{-\alpha}}}} \right){\rm{d}}{r_i}} } \right)}  \hfill\\
 \mathop  = \limits^{\left( c \right)} \exp \left( { - 2\pi {p_c^{D}\lambda _{tu}}{{\left( {s{P_d}} \right)}^{\frac{2}{\alpha }}}}
 \int_0^\infty  {r\int_0^\infty  {{e^{\left( { - t\left( {1 + {r^\alpha }} \right)} \right)}}{\rm{d}}t{\rm{d}}r} }  \right), \hfill\\
 \mathop  = \limits^{\left( d \right)}\exp \left( { - {s^{{2 \mathord{\left/
 {\vphantom {2 \alpha }} \right.
 \kern-\nulldelimiterspace} \alpha }}}  C\left( \alpha  \right) {{p_c^{D}\lambda _{tu}}{P_d}^{{2 \mathord{\left/
 {\vphantom {2 \alpha }} \right.
 \kern-\nulldelimiterspace} \alpha }} } } \right),\hfill\\
 \end{gathered}
\end{equation}
where $(a)$ follows $g_i \sim \exp(1)$, $(b)$ follows from probability generating functional
(PGFL) of PPP \cite{op2} and, $(c)$ results from algebraic manipulation
after converting from Cartesian to polar coordinates and $(d)$ follows from
using some properties of Gamma function, and ${\cal L}_{{I_{{f,ru}}}}(s)$ can be obtained in a similar way.

\subsection{Proof of $R_d$}
For a positive continuous random variable $A$, we can use the following formula for computing its expectation

\begin{equation}\label{E1}
\begin{gathered}
\mathbb{E}\left[ {A\left| {A \ge W} \right.} \right] \hfill\\
  = \int_W^\infty  {t{f_A}\left( t \right)} {\rm{d}}t = \int_W^\infty  {\int_0^t {{f_A}\left( t \right)} {\rm{d}}a} {\rm{d}}t \hfill\\
  = \int_0^W {\int_W^\infty  {{f_A}\left( t \right)} } {\rm{d}}t{\rm{d}}a + \int_W^\infty  {\int_a^\infty  {{f_A}\left( t \right)} } {\rm{d}}t{\rm{d}}a \hfill\\
  = W\Pr \left( {A \ge W} \right) + \underbrace{\int_W^\infty  {\Pr \left( {A \ge a} \right)} {\rm{d}}a}_{S}.\hfill\\
\end{gathered}
 \end{equation}

Next, we focus on the second term of \eqref{E1}, after changing variables with $W=\ln(T_d)$, $A=\ln(1+\gamma_d)$ and $a=\theta_d$ the expression of this term can be given as

\begin{equation}\label{RDproof}
\begin{gathered}
S=\int\limits_{ \ln(T_d)}^\infty  {\Pr \left( {\frac{{{P_d}{h_d}{{\left\| {{X_f}} \right\|}^{ - {\alpha}}}}}{{{I_{d,ru}} + {I_{f,ru}}}} > {e^{{\theta _d}}} - 1} \right)} {\rm{d}}{\theta _d} \hfill\\
\approx \int\limits_{ \ln(T_d)}^\infty  {{L_{{I_{d,ru}}}}\left( {\frac{{{e^{{\theta _d}}}{{\left\| {{X_f}} \right\|}^{{\alpha}}}}}{{{P_d}}}} \right){L_{{I_{f,ru}}}}\left( {\frac{{{e^{{\theta _d}}}{{\left\| {{X_d}} \right\|}^{{\alpha}}}}}{{{P_d}}}} \right)} {\rm{d}}{\theta _d}\hfill\\
= \int\limits_{ \ln(T_d)}^\infty  {\exp \left( { - \pi {{\left\| {{X_d}} \right\|}^2}\beta C\left( {{\alpha }} \right)e^  {\frac{2{\theta _d}}{{{\alpha }}}}} \right)} {\rm{d}}{\theta _d} \hfill\\
= -\frac{{{\alpha}}}{2}{\rm{Ei}}\left( { -T_d^{\frac{2}{\alpha}} \pi {{\left\| {{X_d}} \right\|}^2}\beta C\left( {{\alpha}} \right)} \right), \hfill\\
\end{gathered}
 \end{equation}
where $\beta=\left( {{p_c^{D}\lambda _{tu}} +   {{\left( {{P_f}}/{{P_d}} \right)}^{{2}/{{{\alpha}}}}}{\lambda _f}} \right) $, and the proof is finished.

\subsection{Proof of ${P_f}\left( {{T_f},\alpha} \right)$}
According to the definition of coverage probability, we have
 \begin{equation}\label{PMproof}
 \begin{gathered}
 {P_f}\left( {{T_f},\alpha} \right) = \Pr \left( {\frac{{{P_f}{h_f}{{\left\| {{X_f}} \right\|}^{ - {\alpha}}}}}{{{I_{d,ru}} + {I_{f,ru}}}} \ge {T_f}} \right)\hfill\\
  \quad= \int_0^\infty  {{\rm{Pr}}\left( {{h_f} \ge \frac{{{T_f}r_f^{{\alpha}}}}{{{P_f}}}\left( {{I_{f,ru}} + {I_{d,ru}}} \right)} \right)} {f_{\left\| {{X_f}} \right\|}}\left( {{r_f}} \right){\rm{d}}{r_d} \hfill\\
 \quad\mathop  = \limits^{\left( a \right)} \int_0^\infty  {{{\cal L}_{{I_{f,ru}}}}\left( {\frac{{{T_f}r_f^{{\alpha}}}}{{{P_f}}}} \right){{\cal L}_{{I_{d,ru}}}}\left( {\frac{{{T_f}r_f^{{\alpha}}}}{{{P_f}}}} \right)} {f_{\left\| {{X_f}} \right\|}}\left( {{r_f}} \right){\rm{d}}{r_f} \hfill\\
 \quad \mathop  = \limits^{\left( b \right)} \int_0^\infty  {\exp \left( { - \pi {\lambda _f}r_f^2\rho \left( {{T_f},{\alpha}} \right)} \right)} \hfill\\
  \cdot \exp \left( { - \pi {\lambda _{tu}}r_f^2C\left( {{\alpha}} \right){{\left( {\frac{{{P_d}{T_f}}}{{{P_f}}}} \right)}^{\frac{2}{\alpha}}}} \right)2\pi {\lambda _f}{r_f}{e^{ - \pi {\lambda _f}r_f^2}}{\rm{d}}{r_f}\quad \quad  \hfill\\
 \quad = \frac{1}{{1 + \rho \left( {{T_f},{\alpha}} \right) + \frac{{{\lambda _{tu}}}}{{{\lambda _f}}}C\left( {{\alpha}} \right){{\left( {\frac{{{P_d}{T_f}}}{{{P_f}}}} \right)}^{2/{\alpha}}}}},\hfill\\
 \end{gathered}
 \end{equation}
where (a) follows the setting of $h_f \sim \rm{exp}(1)$ and the independence between inter-tier interference $I_{d,ru}$ and intra-tier interference $I_{f,ru}$, equation (b) follows the definition of the Laplace transform for ${I_{d,ru}}$, ${I_{f,ru}}$, where ${\cal L}_{I_{d,ru}}$ and ${\cal L}_{I_{f,ru}}$ can be derived by using a similar way as \eqref{Pr}, and the proof is finished.

\subsection{Proof of lemma 1}
we first drive the coverage probability of F-AP mode with $\alpha=4$ and $T_f > 1$, which can be denoted as
\begin{equation}\label{RFproof}
\begin{gathered}
 P_f^{\alpha  = 4}\left( {{T_f}} \right) = \frac{1}{{1 + \rho \left( {{T_f},4} \right) + \frac{{{\lambda _{tu}}}}{{{\lambda _f}}}C\left( 4 \right)\sqrt {\frac{{{P_d}{T_f}}}{{{P_d}}}} }} \hfill\\
  = \frac{1}{{1 + \sqrt {{T_f}} \int_{1/\sqrt {{T_f}} }^\infty  {\frac{1}{{1 + {v^2}}}{\rm{d}}v}  + \frac{{\pi {\lambda _{tu}}}}{{2{\lambda _f}}}\sqrt {\frac{{{P_d}{T_f}}}{{{P_f}}}} }} \hfill\\
  \mathop  \approx \limits^{\left( a \right)} \frac{1}{{1 + \sqrt {{T_f}} \left[ {\frac{\pi }{2} - \left( {1/\sqrt {{T_f}} } \right)} \right] + \frac{{\pi {\lambda _{tu}}}}{{2{\lambda _f}}}\sqrt {\frac{{{P_d}{T_f}}}{{{P_f}}}} }} \hfill\\
  = \frac{2}{{{{\pi \sqrt {{T_f}} }}\left( {1 + \frac{{{\lambda _{tu}}}}{{{\lambda _f}}}\sqrt {\frac{{{P_d}}}{{{P_f}}}} } \right)}}, \hfill\\
 \end{gathered}
 \end{equation}
where (a) follows the property of the inverse trigonometric functions that $\arctan(A) \approx A$ if A is smaller than 1, i.e., $T_f \ge 1$.

Then, substituting \eqref{RFproof} into \eqref{E1} with $W=\ln(T_f)$ and we obtain the result.

\subsection{Proof of Theorem 2}
For the system to have a stable equilibrium point, all eigenvalues of the Jacobian of the system of equations should have a negative real part \cite{Evolutionary3}. As we have one system
equation, an equivalence to this condition is that the Jacobian should be negative definite. To this end, we first denote by $f$ the right hand side of \eqref{replicatorSC}. Accordingly, we have

\begin{equation}\label{proofreplicatorCC1}
\begin{gathered}
\frac{\rm{d} f}{\rm{d}x_d^{(D)}}=\sigma x_f^{(D)}\left(1-x_f^{(D)}\right) \left(\frac{\rm{d}\pi_f^{(D)}(\textbf{x})}{\rm{d}x_f^{(D)}}-
\frac{\rm{d}\pi_d^{(D)}(\textbf{x})}{\rm{d}x_f^{(D)}}\right)\;(a)\hfill\\
\quad\quad\quad+\sigma \left(1-2x_d^{(D)}\right)\left( \pi_f^{(D)}(\textbf{x})-\pi_d^{(D)}(\textbf{x})   \right)\quad\quad\;\quad(b)\hfill\\
\end{gathered}
\end{equation}

According to the definition, $\pi_f^{(D)}(\textbf{x})=\pi_d^{(D)}(\textbf{x})$ at the equilibrium point. Hence, $(b)=0$. Then we calculate$ \frac{\rm{d}\pi_f^{(D)}(\textbf{x})}{\rm{d}x_f^{(D)}}$ and $\frac{\rm{d}\pi_d^{(D)}(\textbf{x})}{\rm{d}x_f^{(D)}}$ $(a)$ as \eqref{proofreplicatorCC2} and \eqref{proofreplicatorCC2} on the top of next page

\begin{figure*}[ht]
\begin{equation}\label{proofreplicatorCC2}
\begin{gathered}
\frac{\rm{d}\pi_f^{(D)}(\textbf{x})}{\rm{d}x_f^{(D)}}
=\frac{\rm{d}{\left(\frac{p_fB_f}{N^{(F)}+N^{(D)}x_f^{(D)}}C_f^{(D)}-q_f\phi_f\right)}}{\rm{d}{x_f^{(D)}}}
=-\frac{N^{(D)}p_fB_fC_f^{(D)}}{\left(N^{(F)}+N^{(D)}x_f^{(D)}\right)^2}-c_3c_4N^{(D)}\phi_f\exp\left(c_4\left(N^{(F)}+N^{(D)}x_f^{(D)}\right)\right)
\end{gathered}
\end{equation}
   \hrulefill

\begin{equation}\label{proofreplicatorCC3}
\begin{gathered}
\frac{\rm{d}\pi_d^{(D)}(\textbf{x})}{\rm{d}x_f^{(D)}}
=\frac{\rm{d}{\left(\frac{p_dB_d}{N^{(D)}(1-x_f^{(D)})}C_d^{(D)}-q_d\phi_d\right)}}{\rm{d}{x_f^{(D)}}}
=\frac{N^{(F)}p_dB_dC_d^{(D)}}{\left(N^{(D)}(1-x_f^{(D)})\right)^2}+c_1c_2N^{(D)}\phi_d\exp\left(c_2\left(N^{(D)}(1-x_f^{(D)})\right)\right)
\end{gathered}
\end{equation}
   \hrulefill
\end{figure*}

It is obvious that for any $x_f^{(D)}>0$, we have $\frac{\rm{d}\pi_f^{(D)}(\textbf{x})}{\rm{d}x_f^{(D)}}<0$,and
$\frac{\rm{d}\pi_d^{(D)}(\textbf{x})}{\rm{d}x_f^{(D)}}>0$. Hence this proves $\left(\frac{\rm{d}\pi_f^{(D)}(\textbf{x})}{\rm{d}x_f^{(D)}}-
\frac{\rm{d}\pi_d^{(D)}(\textbf{x})}{\rm{d}x_f^{(D)}}\right)<0$, and $\frac{\rm{d} f}{\rm{d}x_f^{(D)}}$ is strictly negative for
any non-zero value of $x_f^{(D)}$. Therefore, $\dot{x}_f^{(D)}$ evaluated
at any interior equilibrium point is negative. This completes
the proof.

\end{document}